\documentclass[letterpaper,10pt]{IEEEtran}

\usepackage{amsmath,graphicx,epsfig,color,amsfonts}

\graphicspath{{./fig/}}

\usepackage{verbatim}
\usepackage{amssymb}
\usepackage{mathtools}
\usepackage{amsthm}
\usepackage{commath}
\usepackage{nth}
\usepackage{cite}
\usepackage{subfig}
\usepackage[vlined,ruled]{algorithm2e}
\usepackage[capitalise]{cleveref}
\usepackage{xcolor}

\def\expt{\mathbb{E}}
\def\real{\mathbb{R}}
\def\integer{\mathbb{Z}}
\def\natural{\mathbb{N}}
\newcommand{\until}[1]{[#1]}
\newcommand{\subscr}[2]{#1_{\textup{#2}}}
\newcommand{\supscr}[2]{#1^{\textup{#2}}}
\newcommand{\setdef}[2]{\{#1 \; | \; #2\}}
\newcommand{\seqdef}[2]{\{#1\}_{#2}}
\newcommand{\bigsetdef}[2]{\big\{#1 \; | \; #2\big\}}

\newcommand{\union}{\operatorname{\cup}}
\newcommand{\intersection}{\ensuremath{\operatorname{\cap}}}

\newcommand{\ceil}[1]{\left\lceil #1 \right\rceil}

\DeclareMathOperator*{\argmax}{arg\,max}
\DeclareMathOperator*{\argmin}{arg\,min}
\DeclareMathOperator{\tr}{tr}

\newcommand\oprocendsymbol{\hbox{$\square$}}
\newcommand\oprocend{\relax\ifmmode\else\unskip\hfill\fi\oprocendsymbol}

\def \bs {\boldsymbol}
\def \mc {\mathcal}

\newtheorem{theorem}{Theorem}

\newtheorem{lemma}[theorem]{Lemma}

\newtheorem{remark}{Remark}

\newtheorem{assumption}{Assumption}
\newtheorem{definition}{Definition}

\usepackage{caption}
\usepackage{balance}
\title{Multi-Robot Multitask Gaussian Process \\
Estimation and Coverage
	\thanks{This work has been supported in part by NSF grant IIS-1734272 and ARO grant W911NF-18-1-0325.}
}

\author{Lai Wei, Andrew McDonald, and Vaibhav Srivastava
 \thanks{L. Wei is with the Life Sciences Institute,
University of Michigan, Ann Arbor, MI 48109 USA {\tt\small e-mail: weilaitim@gmail.com}} 
	\thanks{V. Srivastava is with the Department of Electrical and Computer Engineering. Michigan State University, East Lansing, MI 48823 USA. {\tt\small e-mail: vaibhav@msu.edu}} 
	\thanks{A. McDonald
		is with the University of Cambridge, Cambridge, UK.
		{\tt\small e-mail: arm99@cam.ac.uk}}%
}

\begin{document}

\maketitle

\newcommand{\rtwo}{\mathbb{R}^2}
\newcommand{\rone}{\mathbb{R}}
\newcommand{\qone}{\mathbb{Q}}
\newcommand{\qtwo}{\mathbb{Q}^2}
\newcommand{\nat}{\mathbb{N}}
\newcommand{\ex}{\mathbb{E}}
\newcommand{\rt}{\rightarrow}
\newcommand{\lt}{\leftarrow}
\newcommand{\T}{\top}
\newcommand{\normal}{\mathcal{N}}
\newcommand{\vor}{\mathcal{V}}
\newcommand{\st}{\mid}
\renewcommand{\abs}[1]{\left|#1\right|}
\renewcommand{\set}[1]{\left\{#1\right\}}
\renewcommand{\norm}[1]{\left\lVert #1\right\rVert}
\newcommand{\obar}[1]{\overline{#1}}
\newcommand{\cov}{\text{cov}}

\newcommand{\bvec}[1]{\bs{#1}}


\begin{abstract}%
\label{sec:abstract}%
Coverage control is essential for the optimal deployment of agents to monitor or cover areas with sensory demands. While traditional coverage involves single-task robots, increasing autonomy now enables multitask operations. This paper introduces a novel multitask coverage problem and addresses it for both the cases of known and unknown sensory demands. For known demands, we design a federated multitask coverage algorithm and establish its convergence properties. For unknown demands, we employ a multitask Gaussian Process (GP) framework to learn sensory demand functions and integrate it with the multitask coverage algorithm to develop an adaptive algorithm. We introduce a novel notion of multitask coverage regret that compares the performance of the adaptive algorithm against an oracle with prior knowledge of the demand functions. We establish that our algorithm achieves sublinear cumulative regret, and numerically illustrate its performance. 
\end{abstract}

\noindent



\section{Introduction}
\label{sec:introduction}
In the parlance of multi-agent systems, coverage control is a critical topic that deals with the optimal deployment of agents to monitor or cover an area. For example, agents may be deployed to cover an entire area to detect and report anomalies or to collect data on parameters like temperature, humidity, or pollution levels across a specified region. While traditional coverage control deals with robots servicing a single task, the increasing autonomous capabilities of robots enable them to handle multiple tasks. For instance, in search and rescue operations, robots might be required to locate survivors, assess structural damage, and deliver supplies concurrently. Similarly, in agricultural applications, robots may need to simultaneously monitor crop health, identify pest infestations, and manage irrigation systems. Motivated by these applications, this paper proposes a new class of coverage problems, namely the multitask coverage problem.

The demand for robot services across different tasks and regions within an environment is often unknown beforehand and must be learned in real-time. This challenge necessitates a principled balance between exploration and exploitation. Agents must explore unknown areas to gather information about service demands and exploit their existing knowledge to efficiently meet these demands. Moreover, service requirements often exhibit spatial correlations, where the demand in one region can influence or predict the demand in neighboring regions. For example, if a particular area in an agricultural field shows signs of pest infestation, nearby areas are also likely to require attention. Similarly, requirements for different services can be correlated; in environmental monitoring, regions with high pollution levels might also show elevated temperatures or humidity. In disaster response, for instance, areas requiring medical assistance due to a high number of injured individuals might also require search and rescue operations to locate and extract trapped survivors. 

To address these complexities, we adopt a multitask Gaussian process (GP) framework that captures these spatial and service correlations. By intertwining multitask learning with multitask coverage, we develop efficient coverage control strategies that can adapt in real-time to the service demands of different regions and tasks.

Classical approaches to coverage control \cite{Cortes2004, Cortes2005} assume that the sensory demand function is known and employ Lloyd's algorithm \cite{Lloyd1982} to guarantee the convergence of agents to a local minimum of the coverage cost. In these algorithms, each agent communicates with the agents in the neighboring partitions at each time and updates its partition. These approaches have been extended to consider obstacle avoidance~\cite{hussein2007effective,panagou2016distributed}, nonuniform coverage~\cite{Lekien2010,leonard2013nonuniform}, nonlinear spaces~\cite{lin2023coverage}, spatiotemporal demands~\cite{diaz2017human,bakolas2013optimal,pratissoli2025distributed,mantovani2026distributed}, energy constraints~\cite{bentz2018complete,munir2024energy}, heterogeneous sensing~\cite{santos2018coverage}, and social fairness~\cite{malencia2023socially}; a review is presented in~\cite{cortes2017coordinated}.

Distributed \emph{gossip-based} coverage algorithms \cite{Bullo2012} address potential communication bottlenecks in classical approaches by updating partitions pairwise between the agents in neighboring partitions. While much of the work in coverage considers continuous convex environments, a discrete graph representation of the environment is considered in~\cite{Durham2012}, which allows for non-convex environments. Additionally, gossip-based coverage algorithms in graph environments converge almost surely to pairwise-optimal partitions in finite time~\cite{Durham2012}. These approaches have been extended to consider one-to-base station communication~\cite{Patel2016} and cloud communication~\cite{peters2017cloud}.

Significant research has also focused on adaptive coverage, where agents do not have prior knowledge of the sensory demand. Parametric estimation approaches to adaptive coverage, such as those by Schwager et al. \cite{Schwager2009, Schwager2017}, model the sensory demand as a linear combination of basis functions and propose algorithms to learn the weight of each basis function. Nonparametric approaches \cite{Choi2008, Xu2011, Luo2018, Luo2019, Todescato2017, Benevento2020, santos2021multi, agrawal2024multi, mantovani2024distributed} model the sensory demand as the realization of a Gaussian Process (GP) and make predictions by conditioning on observed values sampled over the operating environment. Alternative adaptive coverage approaches \cite{Davison2015, Choi2010, zhang2024multi, prajapat2022near} have also been explored. Recently, reinforcement learning and graphical neural network-based approaches have gained traction in adaptive coverage research \cite{gosrich2022coverage, zhang2023constrained,agarwal2025lpac}.

In this paper, we introduce a novel \emph{multitask coverage problem} and design algorithms to address this problem for both known and unknown sensory demands. For the case of known sensory demands, we focus on a federated communication architecture and a discrete environment and design coverage algorithms that converge in finite time. For unknown sensory demands, we adopt a multitask GP framework \cite{bonilla2007multi} to learn the multitask sensory demand functions and integrate it with our coverage algorithm to develop an \emph{adaptive multitask coverage algorithm}. Specifically, we leverage the ``doubling trick" from the multi-armed bandit literature to design a deterministic schedule of learning and coverage.

We characterize the performance of the proposed algorithm using a novel notion of multitask coverage regret, comparing the performance of the adaptive algorithm to that of an oracle that knows the demand functions a priori. We show that the algorithm achieves sublinear cumulative regret. Our notion of regret is distinct from other notions of coverage regret in the literature, such as those in \cite{Benevento2020}, in that it compares the performance of adaptive algorithms with the set of solutions to which coverage algorithms are known to converge, for example, the so-called centroidal Voronoi partitions \cite{cortes2017coordinated}, rather than the globally optimal solution, convergence to which may not be guaranteed.

The major contributions of this work are fivefold. First, motivated by emerging robotic applications, we introduce a novel multitask coverage problem. Second, we design a federated multitask coverage algorithm and establish its convergence properties. Third, we employ a multitask Gaussian Process (GP) framework to learn the sensory demand functions and leverage this to design an adaptive multitask coverage algorithm. Fourth, we introduce a novel notion of multitask coverage regret and demonstrate that our adaptive algorithm achieves sublinear cumulative regret. Finally, we numerically illustrate the proposed algorithms.

An earlier version of this work \cite{LW-AM-VS:20h} was presented at the 2021 International Conference on Robotics and Automation, where we discussed single-task coverage and gossip communication. In this work, we generalize our approach to multitask coverage, which requires the development of novel coverage algorithms and their integration with learning algorithms.

The remainder of the paper is organized as follows. In Section~\ref{sec:background}, we introduce the necessary background information. Section~\ref{sec:problem} presents the multitask coverage problem, the communication model, the multitask estimation framework, and the concept of multitask regret. In Section~\ref{sec:multitask-coverage}, we describe and analyze the federated multitask coverage algorithm. Section~\ref{sec:adaptive-multitask-coverage} extends this algorithm to an adaptive multitask coverage algorithm and establishes its performance in terms of regret. We provide numerical illustrations of the proposed algorithms in Section~\ref{sec:simulations} and conclude the paper in Section~\ref{sec:conclusions}.


\section{Background}\label{sec:background}
\subsection{Graph Representation of Environment}
We consider a discrete environment modeled by an undirected graph $G = (V, E)$, where the vertex set $V$ contains the finite set of points to be covered and the edge set $E  \subseteq V \times V$ is the collection of physically adjacent pairs of vertices that can be reached from each other without passing through other vertices. Let the weight map $w: E \rightarrow \real_{>0}$ indicate the distance between connected vertices. We assume $G$ is connected. Following the standard definition of a weighted undirected graph, a path in $G$ is an ordered sequence of vertices such that there exists an edge between consecutive pairs of vertices. The distance between vertices $v$ and $v'$ in $G$, denoted by $d_G(v,v')$, is defined by the minimum of the sums of the weights in the paths between $v$ and $v'$.

\begin{definition}[$N$-partition]
The $N$-partition of graph $G$ is defined as a collection $\mathcal{P} = \seqdef{P_i}{i=1}^N$ of $N$ nonempty subsets of $V$ such that $\cup_{i=1}^N P_i = V$ and $P_i \cap P_j =\emptyset$ for any $i \neq j$. $\mathcal{P}$ is said to be connected if the subgraph induced by $P_i$ denoted by $G[P_i]$ is connected for each $i \in [N]$. $G[P_i]$ being induced subgraph means its vertex set is $P_i$ and its edge set includes all edges in $G$ whose both end vertices are included in $P_i$.
\end{definition}

\begin{definition}[$N$-covering]
	Given the graph environment $G=(V,E)$, an $N$-covering of $V$ is a collection of $N$ subsets of $V$ denoted by $\mathcal{P} = \seqdef{P_i}{i=1}^N$ such that $\cup_{i=1}^N P_i = V$.
\end{definition}

In contrast to the $N$-partition, $N$-covering allows overlapping subsets $P_i$ and $P_j$.

\subsection{Single Task Coverage Problem}
Consider a team of $N$ robot agents. The configuration of the robot team is a vector of $N$ vertices $\bvec{\eta} \in V^N$ occupied by the robot team, where the $i$-th entry $\eta_i$ corresponds to the position of the $i$-th robot. Suppose there exists a demand function $\phi: V \rt \real_{\geq 0}$ that assigns a nonnegative weight to each vertex in $G$. Intuitively, $\phi(v_i)$ could represent the intensity of a signal or phenomenon of interest at that location, such as the brightness of light, the strength of an acoustic signal, or the concentration of a chemical substance. Given a connected $N$-partition $\mathcal{P}$, the $i$-th robot is tasked to cover the vertices in $P_i$. The coverage cost corresponding to configuration $\bvec{\eta}$ and connected $N$-partition $\mathcal{P}$ can be defined as
\begin{equation}
\label{eq:cost}
H(\bvec{\eta}, \mathcal{P}) = \sum_{i=1}^N \sum_{v'\in P_i} d_{G[P_i]}(\eta_i,v') \phi(v'). 
\end{equation}

In a coverage problem, the objective is to minimize this coverage cost by selecting an appropriate configuration $\bvec{\eta}$ and a connected $N$-partition $\mathcal{P}$. 
There are two intermediate results about the optimal selection of configuration or partition when the other is fixed~\cite{Durham2012}.

For a fixed configuration $\bvec{\eta}$ with distinct entries, an optimal connected $N$-partition $\mathcal{P}$ minimizing coverage cost is called a Voronoi partition denoted by $\mathrm{V}(\bvec\eta)$. Formally, for each $P_i \in \mathrm{V}(\bvec
\eta)$ and any $v' \in P_i$, \[d_G(v',\eta_i) \leq d_G(v',\eta_j), \quad \forall j\in \until{N} .\] 

For a fixed connected $N$-partition $\mathcal{P}$, the centroid of the $j$-th partition $P_j \in \mathcal{P}$ is defined by
\[c_i \in \argmin_{v \in P_i} \sum_{v' \in P_i} d_{G[P_i]}(v, v') \phi(v’),\]
and the optimal configuration is to place one robot at the centroid of each $P_i \in \mathcal{P}$. We denote the vector of the centroid of $\mathcal{P}$ by $\bvec{c}(\mathcal{P})$ with $c_i$ as its $i$-th element.

Building upon the above two properties, the classic Lloyd algorithm iteratively places the robot at the centroid of the current Voronoi partition and computes the new Voronoi partition using the updated configuration. It is known that the robot team eventually converges to a class of partitions called the centroidal Voronoi partition, defined below.

\begin{definition}[Centroidal Voronoi partition,~\cite{DistCtrlRobotNetw}]
An $N$-partition $\mathcal{P}$ is a centroidal Voronoi partition of $G$ if $\mathcal{P}$ is a Voronoi partition generated by some configuration with distinct entries $\bvec{\eta}$, i.e., $\mathcal{P} = \mathrm{V}(\bvec{\eta})$, and 
$\bvec{c}\left(\mathrm{V}(\bvec{\eta})\right) = \bvec{\eta}.$
\end{definition}

It needs to be noted that an optimal partition and configuration pair minimizing the coverage cost $H(\bvec{\eta}, \mathcal{P})$ is of the form $(\bvec{\eta}^*, \mathrm{V}(\bvec{\eta}^*))$, where $\bvec{\eta}^*$ has distinct entries and $\mathrm{V}(\bvec{\eta}^*)$ is a centroidal Voronoi partition. A configuration-partition pair $(\bvec{\eta}', \mathrm{V}(\bvec{\eta}'))$ is considered to be an efficient solution to the coverage problem if $\mathrm{V}(\bvec{\eta}')$ is a centroidal Voronoi partition, even though it is possibly suboptimal~\cite{DistCtrlRobotNetw}.


\section{Problem Formulation}
\label{sec:problem}

In this section, we introduce the multitask coverage problem, the federated communication model, the multitask GP estimation framework, and the multitask coverage regret.  

\subsection{Multitask Coverage with Heterogeneous Agents}

Consider a set of $N$ heterogeneous robotic agents that service $M$ tasks at each vertex $v \in V$. Let $f_i^j: \real \rightarrow \real$ be a strictly increasing function such that $f_i^j(d_G(\eta,v))$ determines the cost of servicing a task $j$ at vertex $v$ by robot $i$ located at vertex $\eta \in V$.  The variation of $f_i^j$ over different agents $i$ and tasks $j$ reflects the heterogeneous capability of agents servicing different tasks. The importance of a vertex $v \in V$ is represented by a nonnegative vector $\Phi(v) = \big[\phi^1(v)\; \cdots\; \phi^{M}(v)\big]^{\top}$ where $\phi^j(v): V \rt \real_{\geq 0}$ reflects the demand of service on the $j$-th task.

Analogously to~\cite{santos2018coverage}, we consider that each vertex may require different agents for different services, and there is one corresponding $N$-covering of environment $\mathcal{P}^j = \seqdef{P_i^j}{i=1}^N$ for each task $j$.  Specifically, $P_i^j \subseteq V$ is the subset of vertices at which agent $i$ services the $j$-th task. With the collection of $N$-coverings for all tasks $\bvec{\mathcal{P}}= \seqdef{\mathcal{P}^j}{j=1}^M$ and robot team configuration $\bvec{\eta} \in V^N$, the multitask coverage cost for the $N$ robotic agents is defined by
\begin{equation}
\label{eq:mulcost}
\mathcal{H}(\bvec{\eta}, \bvec{\mathcal{P}}) = \sum_{j=1}^M \sum_{i=1}^N \sum_{v\in \mathcal{P}_i^j} f_i^j\left(d_G(\eta_i,v)\right) \phi^j(v). 
\end{equation}

We now introduce a few definitions. 
\begin{definition}[Multitask centers]\label{def:multitask-centers}
    Given a collection of $N$-coverings $\bvec{\mathcal{P}}= \seqdef{\mathcal{P}^j}{j=1}^M$, the multitask centers are defined as $\bvec{c}(\bvec{\mathcal{P}}) = \bvec{\eta}^*$ with
 \begin{equation}\label{def:mulcentroid}
		\eta_i^* \in \argmin_{c \in V} \sum_{j=1}^M \sum_{v\in \mathcal{P}_i^j} f_i^j\left(d_G(c,v)\right) \phi^j(v), \forall i\in \until{N}.
		\end{equation}
\end{definition}

\begin{definition}[Multitask equitable partitions]\label{def:multitask-partitions}
Given a robot team configuration $\bvec{\eta}$, the multitask equitable partitions, denoted by
$\mathcal{V}(\bvec{\eta}) = \seqdef{\mathcal{P}^j}{j=1}^M$ consist of $M$ partitions of $V$, consist of $M$ partitions of $V$, where for any task $j \in \until{M}$ and location $v \in P_i^j$,
		\begin{equation}\label{def:mulvoronoi}
			i \in \argmin_{i' \in \until{N}} f_{i'}^j\left(d_G(\eta_{i'},v)\right).
		\end{equation}
\end{definition}
The multitask equitable partitions can be viewed as a conceptual extension of the Voronoi partition, where equation~\eqref{def:mulvoronoi}
assigns every task at each location to the most suitable agent. If multiple robots attain the minimum, ties are broken arbitrarily. 
Note that the collection of $N$-coverings $\bvec{\mathcal{P}}$ being $M$ partitions implies $P^j_{i} \cap P^j_{i'} =\emptyset$ for any  $j \in \until{M}$ and any $i \neq i'$.

We now generalize the idea of a centroidal Voronoi partition~\cite{DistCtrlRobotNetw} to \emph{multitask centroidal equitable partition}. 

\begin{definition}[Multitask centroidal equitable partition] \label{def: MCEP}
For $M$ tasks, the collection of $N$-coverings $\bvec{\mathcal{P}}= \seqdef{\mathcal{P}^j}{j=1}^M$ is a multitask centroidal equitable partition generated by the robot team configuration $\bvec{\eta}\in V^N$ if 
	\begin{itemize}
		\item For any  $j \in \until{M}$, $P^j_{i} \cap P^j_{i'} =\emptyset$ for any $i \neq i'$;
		\item $\bvec{\eta}$ is the vector of multitask centers $\bvec{c}(\bvec{\mathcal{P}})$ associated with $\bvec{\mathcal{P}}$ defined in  equation~\eqref{def:mulcentroid}; 
		\item $\bvec{\mathcal{P}}$ is the collection of $M$ multitask equitable partitions $\mathcal{V}(\bvec{\eta})$ associated with $\bvec{\eta}$ as defined in equation~\eqref{def:mulvoronoi}.   
	\end{itemize}
\end{definition}

\subsection{Federated Communication Model}
Inspired by~\cite{Todescato2017, Patel2016}, we consider a federated (one-to-base station) communication architecture. Our choice is motivated by two reasons. First, a federated architecture allows for a more communication-efficient decentralized nonparametric estimation as noted in~\cite{Todescato2017}. Second, due to the robot heterogeneity, the pairwise optimal partition may not be a centroidal Voronoi partition. Moreover, such an architecture is natural in environments with limited communication, such as underwater environments.  Each robotic agent can communicate with a central base station asynchronously according to the communication model described below.
\begin{itemize}
	\item For each agent, the time between consecutive communications to the base station is lower-bounded by $l>0$ and upper-bounded by $u>l$. We assume the base station communicates with only one agent at a time.	
	\item Each agent $i$ receives and stores a collection of vertices $\seqdef{\mathcal{P}_i}{j=1}^M$ and an appropriate center
 every time it communicates with the base station.
\end{itemize}
We assume the base station has enough power to store and operate on multiple $N$-coverings of $G=(V,E)$ and a list of multitask centroids $\bvec{\eta}\in V^N$.

\subsection{Multitask Nonparametric Estimation}\label{sec:non-parametric-estimation}

Let $\bvec{\Phi} = \big[\Phi^{\top}(v_1)  \; \cdots \; \Phi^{\top}(v_{\abs{V}})\big]^{\top}$ be an $M\abs{V}$ dimensional vector denoting the demands of $M$ different services at all vertices, where $|\cdot|$ denotes set cardinality.
\begin{assumption} \label{ass: prior}
    The demand vector $\bvec{\Phi}$ subjects to a multivariate Gaussian prior distribution $\mathcal{N}( \bvec{\tilde{\mu}}_0 ,\bvec{\tilde{\Sigma}}_0)$, where $\bvec{\tilde{\mu}}_0$ is the prior mean and $\bvec{\tilde{\Sigma}}_0 = \bvec{\Sigma}_0 \otimes \bvec{K}$ is the covariance matrix. A robot can visit any vertex $v_i \in V$ and collect a noisy vector sample $\bvec{y} = \Phi(v_i)+ \bvec{\epsilon}$ corresponding to the demands of $M$ services at $v_i$, where $\bvec{\epsilon} \sim \mathcal{N}(\bvec{0}, \sigma^2 I_M)$.
\end{assumption}
In \cref{ass: prior}, $\bvec{\Sigma}_0$ is the $\abs{V} \times \abs{V}$ covariance matrix reflecting the task similarity between different locations, $\bvec{K}$ is the $M \times M$ positive semidefinite matrix specifying the inter-task correlation, and $\otimes$ denote the Kronecker product. This assumption is consistent with the multitask Gaussian process model defined in~\cite{bonilla2007multi}.


We adopt a Bayesian framework to estimate $\bvec{\Phi}$. Specifically,
 let $n_i(t)$ be the number of samples and $\bvec{s}_i(t)$ be the summation of sampling vectors from $v_i$ until time $t$. Then, from~\cite[Chapter~10]{SMK:93}, the posterior distribution of $\bvec{\Phi}$ at time $t$ becomes $ \mathcal{N} \big( \bvec{\tilde{\mu}}(t) ,\bvec{\tilde{\Sigma}}(t)\big)$, where
\begin{equation}\label{mulposterior}
\begin{split}
\bvec{\tilde{\Sigma}}^{-1}(t) &= \bvec{\tilde{\Sigma}}_0^{-1} + \sum_{i=1}^{\abs{V}} \frac{n_i(t)}{\sigma^2} \bvec{e}_i \bvec{e}_i^{\top} \otimes I_M\\
\bvec{\tilde{\mu}}(t) & = \bvec{\tilde{\Sigma}}(t) \Bigg( \bvec{\tilde{\Sigma}}^{-1}_0 + \sum_{i = 1}^{\abs{V}}\bvec{e}_i \otimes \frac{\bvec{s}_i(t)}{\sigma^2} \Bigg). \\
\end{split}
\end{equation}
Here, $\bvec{e}_i$ is the standard unit vector with $i$-th entry to be $1$.

\subsection{Multitask Coverage Regret}
To achieve efficient coverage, the agents need to balance the trade-off between sampling the environment to learn $\bvec{\Phi}$ (exploration) and 
minimizing the coverage cost using estimated $\bvec{\Phi}$ (exploitation). To characterize this trade-off, we introduce a notion of coverage regret. 

\begin{definition}[Multitask coverage regret]\label{def:mulregret}
	At each time $t$, let the team configuration be $\bvec{\eta}_t$ and the collection of $N$-covering be $\bvec{\mathcal{P}}_t$. The multitask coverage regret with respect to demand function $\bvec{\Phi}$ until time $T$ is defined by $\sum_{t=1}^T \mathcal{R}_t(\bvec{\Phi})$, where 
	\begin{align*}
	\mathcal{R}_t(\bvec{\Phi}) &=  2\mathcal{H}(\bvec{\eta}_t, \bvec{\mathcal{P}}_t) - \mathcal{H}(\bvec{c}(\bvec{\mathcal{P}}_t), \bvec{\mathcal{P}}_t) - \mathcal{H}(\bvec{\eta}_t, \mathcal{V}(\bvec{\eta}_t)).
	\end{align*}
\end{definition}

The instantaneous regret $\mathcal{R}_t(\bvec{\Phi})$ is the sum of two terms $\mc H(\bvec{\eta}_t, \bvec{\mathcal{P}}_t) - \mc H(\bvec{c}(\bvec{\mathcal{P}}_t), \bvec{\mathcal{P}}_t)$ and $\mc H(\bvec{\eta}_t, \bvec{\mathcal{P}}_t) - \mc  H(\bvec{\eta}_t, \vor(\bvec{\eta}_t))$. The former term is the regret induced by the deviation of the current configuration from the optimal configuration for the current partition, while the latter term is the regret induced by the deviation of the current partition from the optimal partition for the current configuration. Accordingly, no regret is incurred at time $t$ if and only if $\bvec{\mathcal{P}}_t$ is multitask equitable partitions and  $\bvec{\eta}_t = \bvec{c}(\bvec{\mathcal{P}}_t)$. Two sources are contributing to the coverage regret. First, the estimation error in the demand function $\bvec{\Phi}$. Second, the deviation from the centroidal equitable partition while sampling the environment to learn $\bvec{\Phi}$ or while the coverage algorithm converges.


\section{Federated multitask Coverage Control}\label{sec:multitask-coverage}

In this section, we propose the federated multitask coverage algorithm for graph environments with known sensory demand functions and establish its convergence properties. Before we present the algorithm, we establish some properties of the multitask coverage function~\eqref{eq:mulcost}. 

\begin{lemma}[Properties of the multitask coverage function]\label{lem:multitask-coverage-funtion}
    The following statements hold for the multitask coverage cost function~\eqref{eq:mulcost}:
    \begin{enumerate}
        \item For a fixed $\bvec{\mathcal{P}}$, $\mathcal{H}(\bvec{\eta}, \bvec{\mathcal{P}})$ is minimized at the multitask center defined by equation~\eqref{def:mulcentroid}; 
        \item For a fixed $\bvec{\eta}$, $\mathcal{H}(\bvec{\eta}, \bvec{\mathcal{P}})$ is minimized 
by the set of multitask equitable partitions defined in Definition~\ref{def:multitask-partitions}.
    \end{enumerate}
\end{lemma}
\begin{proof}
    The statements are easy to verify, and the proof is omitted. 
\end{proof}

To minimize the multitask coverage cost~\eqref{eq:mulcost}, we adapt the single-task coverage algorithm in~\cite{Patel2016} to design the Federated Multitask Coverage Algorithm. Consider a new cost function 
\begin{equation} \label{eq: H_inf}
    \subscr{\mathcal{H}}{inf}(\bvec{\eta}, \bvec{\Phi}) = \sum_{j=1}^M \sum_{v\in V} \min_{i' \in \until{N}} f_{i'}^j\left(d_G(\eta_{i'},v)\right) \phi^j(v).
\end{equation}
Note that in $\subscr{\mathcal{H}}{inf}$, for a given $\bvec{\eta}$, every task at each location is assigned to the optimal robot.

Assume at time $t \in \real_{>0}$, robot $i$ communicates with the base station. The base station computes a new location $\bar{\eta}_{i}$ for robot $i$ such that the new configuration 
\[
\bvec{\eta}_{-i,\bar{\eta}_{i}} = (\eta_1, \ldots, \bar{\eta}_{i},\ldots,\eta_N)
\]
minimizes $\subscr{\mathcal{H}}{inf}$. With the new location $\bar{\eta}_{i}$, for each task, the set of vertices assigned to robot $i$ is updated by adding $\mathcal{P}_{i,+}^j$ and removing $\mathcal{P}_{i,-}^j$ defined by 
		\begin{align*}
		\mathcal{P}_{i,+}^j : =& \bigsetdef{v\in V}{f_{i}^j\left(d_G(\bar \eta_{i},v)\right) <\min_{i' \in \until{N} \setminus \{i\}} f_{i'}^j\left(d_G(\eta_{i'},v)\right)} \\
		\mathcal{P}_{i,-}^j : =& \Big\{v\in \mathcal{P}_i^j \intersection (\cup_{i \neq i'} \mathcal{P}_{i'}^j) \mid  f_{i}^j\left(d_G(\eta_{i},v)\right) \\
		&\qquad \qquad  \geq \min_{i' \in \until{N} \setminus \{i\}} f_{i'}^j\left(d_G(\eta_{i'},v)\right)\Big\}.
		\end{align*}

Note that $\mathcal{P}_{i,+}^j$ includes all the vertices whose $j$-th task is best served by robot $i$ at location $\bar{\eta}_i$ and $\mathcal{P}_{i,-}^j$ is the set of vertices where the $j$-th task is assigned to multiple robots, including robot $i$, and its service is not superior to other assigned robots. The details of the algorithm are described in Algorithm~\ref{algo:FMC}.

\begin{algorithm}[ht!]	
	\smallskip
	{\footnotesize  
		\SetKwInOut{Input}{  Input}
		\SetKwInOut{Set}{  Set}
		\SetKwInOut{Title}{Algorithm}
		\SetKwInOut{Require}{Require}
		\SetKwInOut{Output}{Output}
		
		\Input{Environment graph $G$ and demand function $\bvec{\Phi}$ \;}
		
		\Set{$\alpha \in (0,1)$ and $\beta>1$\;}
		
		\medskip
		
		\nl\If{$\subscr{\mathcal{H}}{inf}(\bvec{\eta}) > \min_{v\in V} \subscr{\mathcal{H}}{inf}(\bvec{\eta}_{-i,v} )$}{
			\smallskip
			Update $\eta_i$ in $\bvec{\eta}$ such that
			\[\eta_i \in \argmin_{v\in V} \subscr{\mathcal{H}}{inf}(\bvec{\eta}_{-i,v} ). \]			
		}
		  \nl For each task $j\in \until{M}$, compute sets for robot $i$:
		
		\begin{align*}
		{P}_{i,+}^j : =& \bigsetdef{v\in V}{f_{i}^j\left(d_G(\eta_{i},v)\right) <\min_{i' \in \until{N}} f_{i'}^j\left(d_G(\eta_{i'},v)\right)}, \\
		{P}_{i,-}^j : =& \Big\{v\in {P}_i^j \intersection (\cup_{i \neq i'} P_{i'}^j) \mid  f_{i}^j\left(d_G(\eta_{i},v)\right) \\
		&\qquad \qquad \qquad \qquad \geq \min_{i' \in (\until{N} \setminus \{i\})} f_{i'}^j\left(d_G(\eta_{i'},v)\right)\Big\}. 
		\end{align*}
		
		Set ${P}_i^j \leftarrow ({P}_i^j \setminus {P}_{i,-}^j) \cup P_{i,+}^j$ for each $j\in\until{M}$.
		
		\nl Inform agent $i$ its new position $\eta_i$ and task-specific sets $\seqdef{{\mathcal{P}}_i^j}{j=1}^M$.
		
		\smallskip

		\caption{Federated Multitask Coverage}
		\label{algo:FMC}
	}
\end{algorithm}

To analyze the federated multitask coverage algorithm, we intend to apply a Lyapunov argument that relies on the following Lemma regarding the convergence of dynamic systems driven by a set-valued map $T: X \rightrightarrows X$.

\begin{lemma}[Lemma 3,~\cite{Patel2016}]\label{lemma: svm_conv}
	Let $(X, d)$ be a finite metric space. Given a collection of maps $T_1,\ldots,T_m: X\rightarrow X$, define the set-valued map $T: X \rightrightarrows X$ by $T(x) = \{T_1(x),\ldots,T_{m}(x)\}$ and let $\seqdef{x_n}{n\in \integer_{\geq 0}}$ be a sequence such that $x_{n+1} = T_{\nu(n)}(x_n)$, where $\nu: \integer_{\geq 0} \rightarrow \until{m}$. Assume that:
	\begin{enumerate}
		\item there exists a Lyapunov function $U: X \rightarrow \real$ such that $U(x') < U(x)$ for all $x \in X$ and $x' \in T(x) \setminus \{x\}$; and
		
		\item for all $i \in  \until{m}$, there exists a subsequence of time steps sequence $\seqdef{n_k}{k\in \integer_{\geq 0}}$ such that $\sigma(n_k) = i$ and $n_{k+1} - n_{k}$ is bounded.
	\end{enumerate}
Let $F_i = \setdef{x\in X}{T_i(x) = x}$ be the set of fixed point of $T_i$. Then, for all $x_0 \in X$, there exists $n'$ and $\bar x \in (F_1 \cap \ \cdots \ \cap F_m)$ such that $x_n = \bar{x} $ for all $n\geq n'$.
\end{lemma}

We now establish the convergence properties of Algorithm~\ref{algo:FMC} as follows. 

\begin{lemma}[Convergence of federated multitask coverage algorithm]\label{lem:convergence-coverage}
For the multitask coverage problem with heterogeneous agents, the federated communication protocol, and the known demand function $\bvec{\Phi}$, the federated multitask coverage algorithm converges to the set of multitask centroidal equitable partitions in a finite number of steps. 

\end{lemma}
\begin{proof} The proof proceeds in three phases, characterizing convergence of the configuration, the coverings, and removal of overlaps. For each phase, we construct a Lyapunov function $U_1, U_2$, and $U_3$, respectively, to analyze the system evolution.

\noindent \textbf{Step 1 (Convergence of configuration):}
In the federated multitask coverage algorithm, the update of configuration $\bvec{\eta}$ can be written as a set-valued map $T^1: V^N \rightrightarrows V^N$ such that
\[T^1_{i}(\bvec{\eta}) = \bvec{\eta}_{-i,\bar{\eta}_{i}}, \; 
\forall i\in \until{N}.\]
where $\bar{\eta}_i$ is the new location for robot $i$ assigned by~\cref{algo:FMC} when robot $i$ communicates with the base station. Taking function $U_1(\bvec{\eta})  = \subscr{\mathcal{H}}{inf}(\bvec{\eta},\bvec{\Phi})$ defined in~\eqref{eq: H_inf}, we apply~\cref{lemma: svm_conv} to get the configuration $\bvec{\eta}_t$ converge to a local minimum $\bvec{\eta}^*$ of $\subscr{H}{inf}$ in finite time.

\noindent \textbf{Step 2 (Convergence of $N$-coverings):} After the robot team configuration converges to the local minimum $\bvec{\eta}^*$, it remains unchanged. We then consider the evolution of the collection of $N$-coverings for the robot team, $\bvec{\mathcal{P}}= \seqdef{\mathcal{P}^j}{j=1}^M$. We show that, in finite time, $\bvec{\mathcal{P}}$ converges to a point $\bvec{\hat{\mathcal{P}}}$, where each task $j \in \until{M}$ at each location $v\in V$ is assigned to an optimal robot. Specifically, if $v \in \hat{P}_i^j$, then
\begin{equation}\label{eq: int_opt_cov}
    i \in \argmin_{i' \in \until{N}} f_{i'}^j\left(d_G(\eta^*_{i'},v)\right) \phi^j(v): v \in \hat{P}_{i'}^j.
\end{equation}

With the step $2$ in~\cref{algo:FMC}, we define a set-valued map $T^2$ such that
\begin{equation}\label{eq: T^2}
    T^2_{i}(\bvec{\mathcal{P}}) = \{{P}^j_{1},...,\bar{{P}}^j_{i} ,...,{P}^j_{N}\}_{j=1}^M, \forall i\in \until{N},
\end{equation}
where ${\bar{P}}_i^j = ({P}_i^j \setminus {P}_{i,-}^j) \cup P_{i,+}^j$ after robot $i$ communicates with the base station. Note that if $\bvec{\mathcal{P}}^j$ is an $N$-covering, $\{{P}^j_{1},...,\bar{{P}}^j_{i} ,...,{P}^j_{N}\}$ is also an $N$-covering.  Since task $j$ at location $v\in V$ can be assigned to multiple robots, we only consider the robot $i$ that minimizes servicing cost $f_i^j\left(d_G(\eta^*_i,v)\right) \phi^j(v)$, and define
\[U_2(\bvec{\mathcal{P}}) = \sum_{j=1}^M \sum_{v\in V} \min_{i \in \setdef{i'\in \until{N}}{v\in P_i^j}} f_i^j\left(d_G(\eta^*_i,v)\right) \phi^j(v). \]
With the definition of $P_{i,+}^j$, we apply Lemma~\ref{lemma: svm_conv} to get the configuration $\bvec{\mathcal{P}}$ converges to a $\bvec{\hat{\mathcal{P}}}$ satisfying~\eqref{eq: int_opt_cov} in finite time.

For each task $j$, we define
as the set of vertices that are owned by more than one agent.


\noindent \textbf{Step 3 (Removal of overlap):} After $\bvec{\mathcal{P}}$ converges to a $\bvec{\hat{\mathcal{P}}}$, the task $j$ at some location $v\in V$ can still be assigned to multiple robots. Recall that for a  multitask centroidal equitable partition in~\cref{def: MCEP}, each task $j$ at each location $v\in V$ can be allocated to only one robot. In this step, we show  $\bvec{\mathcal{P}}_t$ converges from $\bvec{\hat{\mathcal{P}}}$ to a multitask centroidal equitable partition $\hat{\hat{\bvec{\mathcal{P}}}}$ in finite time. Toward this end, we define
\[
U_3(\bvec{\mathcal{P}}) = \sum_{j=1}^M \sum_{v \in V} \abs{\setdef{i \in [N]}{v \in P_i^j} }.\]
With the set-valued map $T^2$ defined in~\eqref{eq: T^2} and the definition of $P_{i,-}^j$, the claim is a direct consequence of~\cref{lemma: svm_conv}.

With the above $3$ finite-time convergence result, the robot team is shown to converge to a multitask centroidal equitable partition in finite time.
\end{proof}


\section{Adaptive Multiple Task Coverage Control}
\label{sec:adaptive-multitask-coverage}

In this section, we present and analyze an adaptive multitask coverage algorithm called the Deterministic Sequencing of Multitask Learning and Coverage (DSMLC) algorithm. 

\subsection{The DSMLC Algorithm}

To address the estimation and coverage of multiple demands with heterogeneous agents, we design the DSMLC algorithm, which leverages the multitask nonparametric estimation from Section~\ref{sec:non-parametric-estimation} and the federated multitask coverage algorithm. It maintains the epoch-wise structure composed of exploration, information propagation, and coverage phases. 

Recall the multitask GP update equation~\eqref{mulposterior}. Let $\bvec{\tilde{\Sigma}}_i(t)$ be the $M\times M$ diagonal sub-block of the covariance matrix $\bvec{\tilde{\Sigma}}(t)$ associated with $\Phi(v_i)$. We first design a greedy policy to reduce the uncertainty in the multitask GP posterior.  

Let $X_n = (v_{s_1},\ldots,v_{s_n})$ be a sequence of $n$ vertices selected by the DSMLC algorithm  and let $\bvec{Y}_{n} = (\bvec{y}_1,\ldots, \bvec{y}_n)$ be $n$ observed $M$-dimensional noisy vectors corresponding to the demands at $(v_{s_1},\ldots,v_{s_n})$. Recall that, for $\bvec{Y} \sim \mathcal{N}(\bvec{\mu}, \bvec{\mathrm{\Sigma}})$, the entropy $\mathrm{H}\left( \bvec{Y} \right)  = \frac{1}{2} \log \abs{2 \pi e \bvec{\mathrm{\Sigma}}}$. With $ \mathrm{H}\left( \bvec{Y}_{n} \mid \bvec{\Phi}\right) = \frac{1}{2} \log \abs{2\pi e I_{M}}$, we apply chain rule of joint entropy to get the mutual information of $\bvec{Y}_{n}$ and $\bvec{\Phi}$ to be expressed as
\begin{align}
\mathrm{I}_{X_n} \left(\bvec{Y}_{n}; \bvec{\Phi}  \right) =& \mathrm{H}\left( \bvec{Y}_{n} \right) - \mathrm{H}\left( \bvec{Y}_{n} \mid \bvec{\Phi}\right) \nonumber\\
= & \sum_{k=1}^n \mathrm{H}\left( \bvec{Y}_{k} \mid \bvec{Y}_{k-1}  \right) -  \frac{1}{2} \log \abs{2\pi e \sigma^2 I_{M}} \nonumber\\
= & \frac{1}{2} \sum_{k=1}^{n} \log \abs {I_M + \sigma^{-2} \bvec{\tilde{\Sigma}}_{i}(k-1)}, \label{def: multmi}
\end{align}
where $\bvec{\tilde{\Sigma}}_{i}(k-1)$ is the marginal posterior variance of $\Phi(v_i)$ given $\bvec{Y}_{k-1}$. Using~\eqref{def: multmi}, our greedy policy maximizes one-step mutual information and selects the most uncertain vertex 
\begin{align}\label{eq:greedy}
s_{k}= \argmax_{i \in V} \abs {I_M + \sigma^{-2} \bvec{\tilde{\Sigma}}_{i} (k-1)}.    
\end{align}

Suppose the trace $ \tr\big(\bvec{\tilde{\Sigma}}_i(0)\big) \leq \tau$ for all $i$. Then, at the beginning of the $\ell$-th exploration epoch, the base station recursively computes the sampling points with~\eqref{eq:greedy} until
\[
\max_{i\in V}\tr\big(\bvec{\tilde{\Sigma}}_i(k) \big) < \alpha^{\ell}\tau,
\]
for some tunable constant $\alpha \in (0,1)$. Note that since the covariance update in~\eqref{mulposterior} only requires sampling location and not the actual measurement, the above computation can be done before the actual sampling. 

After computing the sampling points, the base station assigns each sampling point to the robot associated with the nearest location in the current multitask center configuration $\bvec{c}(\bvec{\mathcal{P}}_t)$. Each robot then adopts an appropriate vehicle routing algorithm~\cite{toth2002vehicle} to traverse the assigned sampling points.

In the information propagation phase, the robots transmit sufficient statistics for the update in equation~\eqref{mulposterior} to the base station, which then determines $\hat{\bvec{\Phi}}=\bvec{\tilde{\mu}}(t)$. 

Finally, in the $\ell$-th coverage phase, robots implement the federated MH coverage algorithm with the estimated demand function $\hat{\bvec{\Phi}}$ for a duration $\lceil \beta^\ell \rceil$, for some tunable $\beta >1$.

The detailed execution of these phases of DSMLC is shown in Algorithm~\ref{algo:MHDSLC}. 

\begin{algorithm}[t!]	
	\smallskip
	{\footnotesize  
		\SetKwInOut{Input}{  Input}
		\SetKwInOut{Set}{  Set}
		\SetKwInOut{Title}{Algorithm}
		\SetKwInOut{Require}{Require}
		\SetKwInOut{Output}{Output}
		
		\Input{Environment graph $G$, $(\tilde{\bvec{\mu}}_0$, and $\bvec{\tilde \Sigma}_0)$ \;}
		
		\Set{$\alpha \in (0,1)$ and $\beta>1$\;}
		
		\medskip
		
		\For{epoch $\ell =1,2,\hdots$}{
			\smallskip
			\emph{\textbf{Exploration phase: }}
			
			\nl Select sampling points recursively using~\eqref{eq:greedy} until
			\[\max_{i\in V} \tr\big(\bvec{\tilde{\Sigma}}_i(k) \big)  \leq  \alpha^\ell \tau. \]

            \nl Assign each sampling point to the robot whose multitask center is closest. 
   
			\emph{\textbf{Information propagation phase: }}
			
			\nl Each robot agent sends the sensing results to the base station
			
			\nl Estimate demand function $\hat{\bvec{\Phi}}=\bvec{\tilde{\mu}}(t)$ using~\eqref{mulposterior}
			
			\smallskip
			
			\emph{\textbf{Coverage phase: }}
			
			\nl \For{$t_\ell=1,2,\ldots, \ceil{\beta^\ell}$}{\ Based on the current estimate $\hat{\bvec{\Phi}}$, follow~\cref{algo:FMC} to update robot team configuration and partitions.}

		}
		\caption{DSMLC}
		\label{algo:MHDSLC}
	}
\end{algorithm}

\subsection{Analysis of the DSMLC algorithm}

In this section, we analyze DSMLC to provide a performance guarantee about the expected cumulative coverage regret. To this end, we leverage the information gained from the estimation phase to analyze the convergence rate of uncertainty. Then, we leverage the convergence properties of the federated MH coverage algorithm to establish an upper bound on the expected cumulative coverage regret.

\subsubsection{Mutual Information}

Recall that for the multitask Gaussian process defined in Section~\ref{sec:non-parametric-estimation}, the covariance function is $\bvec{\tilde{\Sigma}}_0 = \bvec{\Sigma}_0 \otimes \bvec{K}$. We first focus on the single task case with covariance function $\bvec{\Sigma}_0$ and then generalize to the case of multiple tasks.

For a single-task GP $\bvec{\phi}$, let $X_n \in V^n$ be a sequence of $n$ vertices and $\supscr{Y}{sngl}_n \in \real^n$ be observed sampling results corresponding to $X_n$.  If the sampling noise has variance $\sigma^2$, let $\mathrm{I}_{X_n} \left(\supscr{Y}{sngl}_n ; \bvec{\phi} \right) = \frac{1}{2}\log |(I_n + \sigma^{-2} \Sigma_{X_n})|$ be the mutual information.
Then, the $n$-step single-task maximum information gain is defined by 
\[
\supscr{\gamma}{sgnl}_n (\sigma^2) = \max_{\supscr{X}{sngl}_n \in V^n } \mathrm{I}_{X_n} \left(\supscr{Y}{sngl}_n ; \bvec{\phi}\right).
\]

Typically, it is hard to characterize $\supscr{\gamma}{sgnl}_n $ with a general $\bvec{\Sigma}_0$. Therefore, we make the following assumption.

\begin{assumption}\label{assum: ig}
	Vertices in $V$ lie in a convex and compact set $D \in \real^2$ and the covariance $\bvec{\Sigma}_0$ in~\cref{ass: prior} is determined by an exponential kernel function 
	\begin{equation}\label{def: stker}
	k(v_i, v_j) = \sigma_v^2 \exp\left(-\frac{\subscr{d}{eu}^2(v_i,v_j)}{2 l^2} \right),
	\end{equation}
	where $\subscr{d}{eu}(v_i,v_j)$ is the Euclidean distance between $v_i$ and $v_j$, $l$ is the length scale, and $\sigma_v^2$ is the variability parameter. 
\end{assumption} 

 We now recall an upper bound on $\gamma_n$ from~\cite[Th. 5]{NS-AK-SMK-MS:12}.
\begin{lemma}[{Information gain for squared exp. kernel}] \label{lemma: parmi}
     Let $D \subset \real^d$ be compact and convex, $d \in \natural$. Assume that the kernel function satisfies $k(x, x') \leq 1$. The maximum mutual information for the squared exponential kernel satisfies $\supscr{\gamma}{sngl}_n\in O\big( (\log n)^{d+1} \log\big(1 + \sigma_v^2/\sigma^2\big)\big)$.
\end{lemma}

For the multitask Gaussian process $\bvec{\Phi}$, given a sequence of vertices $X_n$ and the corresponding multitask observations $\bvec{Y}_n$ with sampling noise variance $\sigma^2$, we define the $n$-step multitask maximal mutual information gain with inter-task covariance matrix $\bvec{K}$ as
\begin{equation}
    \gamma_n^{\bvec{K}}(\sigma^2) := \max_{X_n \in V^n } \mathrm{I}_{X_n}\left(\bvec{Y}_{n} ; \bvec{\Phi}\right). \label{def: max_mt_info}
\end{equation}

We now establish a relationship between $\gamma_n^{\bvec{K}}$ and $\supscr{\gamma}{sngl}_n$.
 
\begin{lemma}[Multitask Maximal Mutual Information Gain]\label{lem:multitask-info-gain}
Given the $n$-step single task maximum information $\supscr{\gamma}{sngl}_n$ associated with covariance function $\bvec{\Sigma}_0$ and the inter-task covariance matrix $\bvec{K}$, the $n$-step multitask maximum information gain $ \gamma_n^{\bvec{K}}$ associated with covariance function $\bvec{\tilde{\Sigma}}_0$ satisfies 
	\[ 
 \gamma_n^{\bvec{K}} (\sigma^2) \leq \sum_{j=1}^M \supscr{\gamma}{sngl}_{n}(\sigma^2/\supscr{\lambda_j}{tsk}),
 \]
 where $\supscr{\lambda_1}{tsk},\ldots,\supscr{\lambda_M}{tsk}$ are the eigenvalues of $\bvec{K}$.
\end{lemma}

\begin{proof}
    Let $\mathrm{vec}(\bvec{Y_n})$ be the $M\times n$ dimensional vector constructed by concatenating $\bvec{y}_1,\ldots, \bvec{y}_n$. Then, the covariance matrix of $\mathrm{vec}(\bvec{Y_n})$ becomes $\bvec{\Sigma}_{X_n} \otimes \bvec{K}$, where the $(i,j)$-th element equals the $(s_i,s_j)$-th element of $\bvec{\Sigma}_0$.
    \begin{align*}
        \mathrm{I}_{X_n} \left(\bvec{Y}_{n}; \bvec{\Phi}  \right) & =  \frac{1}{2} \log \abs{I_{Mn} + \sigma^{-2}\bvec{\Sigma}_{X_n} \otimes \bvec{K}}
    \end{align*}
    Let $\lambda_1,\ldots,\lambda_n$ be the eigenvalues of $\bvec{\Sigma}_{X_n}$.
    According to the property of the Kronecker product, 
    \begin{align*}
        \mathrm{I}_{X_n} \left(\bvec{Y}_{n}; \bvec{\Phi}  \right) &= \sum_{j=1}^{M} \sum_{i=1}^{n} \log (1+\sigma^{-2} \lambda_i \supscr{\lambda_j }{tsk}) \\
        &= \sum_{j=1}^{M} \frac{1}{2} \log \abs{I_{n} + \sigma^{-2} \supscr{\lambda_j }{tsk} \bvec{\Sigma}_{X_n}}. 
    \end{align*}
Thus, 
    \begin{align*}
       \gamma_n^{\bvec{K}} & = \max_{X_n}   \;   \mathrm{I}_{X_n} \left(\bvec{Y}_{n}; \bvec{\Phi} \right) \\
      &=  \sum_{j=1}^{M} \frac{1}{2} \log \abs{I_{n} + \sigma^{-2} \supscr{\lambda_j }{tsk} \bvec{\Sigma}_{X_n}} \\
      & \le \! \sum_{j=1}^{M} \max_{X_n}  \frac{1}{2} \log \abs{I_{n} + \sigma^{-2} \supscr{\lambda_j }{tsk} \bvec{\Sigma}_{X_n}} \! = \!\sum_{j=1}^M  \supscr{\gamma}{sngl}_{n}(\sigma^2/\supscr{\lambda_j}{tsk}).
    \end{align*}
\end{proof}
\begin{remark}
    Substituting Lemma~\ref{lemma: parmi} into Lemma~\ref{lem:multitask-info-gain}, the inter-task covariance matrix $\bvec{K}$ introduces a scaling factor $\sum_{j=1}^M  \log\big(1 + \supscr{\lambda_j}{tsk}\sigma_v^2/\sigma^2\big)\big)$ on the multitask maximal mutual information gain $\gamma_n^{\bvec{K}} (\sigma^2)$, which decreases as the correlation between tasks increases.
\end{remark}


\subsubsection{Uncertainty Reduction}

We now present a bound on the trace of multitask covariance at each node after sampling at vertices within $X_n$. Recall that the demand vector $\bvec{\Phi}$ has a multivariate Gaussian prior distribution $\mathcal{N}( \bvec{\tilde{\mu}}_0 , \bvec{\Sigma}_0 \otimes \bvec{K})$.

\begin{lemma}[Uncertainty reduction]\label{lemma: ur}
        Let $\bar{\sigma}_0^2$ be the maximum diagonal term of $\bvec{\Sigma}_0 $ and  $\supscr{\lambda_1}{tsk} \geq \supscr{\lambda_2}{tsk} \ldots \geq \supscr{\lambda_M}{tsk}$ be the eigenvalues of $\bvec{K}$. Under the greedy sampling policy, the trace of the multitask covariance matrix after $n$ sampling rounds satisfies
	\begin{equation*}
	\max_{i \in V} \tr\big(\bvec{\tilde{\Sigma}}_i(n) \big)  \leq \frac{2\bar{\sigma}_0^2 \supscr{\lambda_1}{tsk}}{\log \big(1 + \sigma^{-2} \bar{\sigma}_0^2 \supscr{\lambda_1}{tsk} \big)} \frac{\gamma^{\bvec{K}}_n}{n}.
	\end{equation*}
\end{lemma}

\smallskip 
\begin{proof} 
    For any $i \in V$, $ \big| {I_M + \sigma^{-2} \bvec{\tilde{\Sigma}}_{i}(k)}\big|$ is non-increasing with $k$. Hence, with $s_k$ defined in~\eqref{eq:greedy}, we have
    \begin{equation}\label{ineq: variance}
	\begin{split}
	\abs {I_M + \sigma^{-2} \bvec{\tilde{\Sigma}}_{s_{k+1}}(k)} &\leq \abs {I_M + \sigma^{-2} \bvec{\tilde{\Sigma}}_{s_{k+1}}(k-1)} \\
	&\leq \max_{i \in V} \abs {I_M + \sigma^{-2} \bvec{\tilde{\Sigma}}_{i}(k - 1)}\\
        & = \abs {I_M + \sigma^{-2} \bvec{\tilde{\Sigma}}_{s_{k}}(k-1)}
	\end{split}	
    \end{equation}
    which indicates that $\abs {I_M + \sigma^{-2} \bvec{\tilde{\Sigma}}_{s_{k}}(k-1)}$ is monotonically non-increasing with $k$. Hence, from~\eqref{def: multmi} and~\eqref{def: max_mt_info},
    \begin{equation} \label{eq:inf_bound}
        \log \abs {I_M + \sigma^{-2} \bvec{\tilde{\Sigma}}_{s_n}(n-1)}  \leq 2\gamma^{\bvec{K}}_n/n.
    \end{equation}
    
    For any $i \in V$, let $\lambda_i^1(k)\geq \lambda_i^2(k) \ldots \geq \lambda_i^M (n)$  be the eigenvalues of $\bvec{\tilde{\Sigma}}_{i}(n)$. Then we have
    \begin{equation} \label{eq: info}
        \log \abs {I_M + \sigma^{-2} \bvec{\tilde{\Sigma}}_{i}(n)} = \sum_{j =1}^M \log \big(1 + \sigma^{-2} \lambda_{i}^j (n) \big).
    \end{equation}
    Since ${x}/{\log \left(1+x\right)}$ is an increasing function for $ x \in [0,\infty)$, 
    \begin{align*}
        \lambda_{i}^j (n) \leq \frac{\bar{\sigma}_0^2 \supscr{\lambda_1}{tsk}}{\log \big(1 + \sigma^{-2} \bar{\sigma}_0^2 \supscr{\lambda_1}{tsk} \big)} \log \big(1 + \sigma^{-2} \lambda_{i}^j (n) \big),
    \end{align*}
    where we apply the fact $\lambda_i^M (k) \leq \lambda_i^M (0) \leq \bar{\sigma}_0^2 \supscr{\lambda_1}{tsk}$ for any $i \in V$. Putting the above inequality together with~\eqref{eq: info},
    \begin{align*}
        &\max_{i \in V} \tr\big(\bvec{\tilde{\Sigma}}_i(n) \big) = \max_{i \in V} \sum_{j=1}^M \lambda_{i}^j (n)\\
        \leq & \frac{\bar{\sigma}_0^2 \supscr{\lambda_1}{tsk}}{\log \big(1 + \sigma^{-2} \bar{\sigma}_0^2 \supscr{\lambda_1}{tsk} \big)}  \max_{i \in V} \log \abs {I_M + \sigma^{-2} \bvec{\tilde{\Sigma}}_{i}(n)} \\
         \leq & \frac{\bar{\sigma}_0^2 \supscr{\lambda_1}{tsk}}{\log \big(1 + \sigma^{-2} \bar{\sigma}_0^2 \supscr{\lambda_1}{tsk} \big)}  \max_{i \in V} \log \abs {I_M + \sigma^{-2} \bvec{\tilde{\Sigma}}_{i}(n - 1)},
    \end{align*}
    where the last step is due to~\eqref{ineq: variance}. Substituting~\eqref{eq:inf_bound} into the above equation, we conclude the proof
\end{proof}

\begin{remark}
If the correlation information is ignored, i.e., $\Phi(v)$ for each $v \in V$ are treated as independent, it follows that under a greedy sampling policy $\max_{i \in V} \tr\big(\bvec{\tilde{\Sigma}}_i(n) \big) \in O (\abs{V}/n)$. In contrast, when the correlation structure is taken into account, substituting the results of Lemmas~\ref{lemma: parmi} and~\ref{lem:multitask-info-gain} into Lemma~\ref{lemma: ur} yields $\max_{i \in V} \tr\big(\bvec{\tilde{\Sigma}}_i(n) \big) \in O ((\log (n))^3/n)$, which eliminates the dependence on $|V|$. This highlights the substantial advantage of exploiting spatial correlations, especially when the environment is finely discretized.
\end{remark}

\subsubsection{An Upper Bound on Expected Multitask Coverage Regret}

 We now present an upper bound on the expected multitask coverage regret for the DSMLC algorithm. 

\begin{theorem}\label{theorem:regret}
	For DSMLC and any time horizon $T$, if Assumption~\ref{assum: ig} holds and $\alpha= \beta^{-2/3}$, then the expected cumulative coverage regret with respect demand function $\bvec{\Phi}$ satisfies
	\[ 
 \expt \Bigg[ \sum_{t=1}^{T} \mathcal{R}_t(\bvec{\Phi})\Bigg]\in O \big (T^{2/3} (\log(T))^3 \big).
 \]
\end{theorem}

\medskip

\begin{proof}
	We establish the theorem using the following four steps. 
	
	\noindent \textbf{Step 1 (Regret from estimation phases):}
	Let the total number of sampling steps before the end of the $j$-th epoch be $s_j$. By applying Lemmas~\ref{lemma: parmi}-\ref{lemma: ur}, we get 
	\[s_j  \in  O ({(\log(T))^3}/{\alpha^j}).\]
Thus, the coverage regret in the estimation phases until the end of the $j$-th epoch belongs to $O ({(\log(T))^3}/{\alpha^j})$. 
	
	\noindent \textbf{Step 2 (Regret from information propagation phases):}
	The sampling information by each robot is transmitted to the base station. Under the federated communication model, the base station receives information from all robots in a finite number of communication rounds.  
 Thus, before the end of the $j$-th epoch, the coverage regret from information propagation phases can be bounded by $c_1 j$ for some constant $c_1>0$.	
	
	\noindent \textbf{Step 3 (Regret from coverage phases):}
	According to Lemma~\ref{lem:convergence-coverage}, in each coverage phase, the expected time before converging to a multitask centroidal equitable partition is finite. Thus, before the end of the $j$-th epoch, the expected coverage regret from converging steps can be upper bounded by $c_2 j$ for some constant $c_2>0$.
	
	Also, note that the robot team converges to 
 a multitask centroidal equitable partition for estimated demand function $\hat{\bvec{\Phi}}$ which may deviate from the actual $\bvec{\Phi}$. 
 Since $H$ is linear in the demand function $\bvec{\Phi}$, the instantaneous coverage regret $\mathcal{R}_t(\bvec{\Phi})$ 
can be expressed as a linear functional 
 \begin{align*}
2\mathcal{H}(\bvec{\eta}_t, \bvec{\mathcal{P}}_t) - \mathcal{H}(\bvec{c}(\bvec{\mathcal{P}}_t), \bvec{\mathcal{P}}_t) - \mathcal{H}(\bvec{\eta}_t, \mathcal{V}(\bvec{\eta}_t)):= A_t^{\top} \bvec{\Phi},
	\end{align*}
	for some $A_t \in \real^{M\abs{V}}$. Moreover, the posterior distribution of $\mathcal{R}_t(\bvec{\Phi})$ can be written as $\mathcal{N} (A_t^{\top} \bvec{\tilde \mu}(t), A_t^{\top} \bvec{\tilde \Sigma}(t) A_t) $.

 Since $\hat{\bvec{\Phi}} = \bvec{\tilde \mu}(t)$, and at the associated multitask centroidal equitable partition, the multitask coverage regret should be zero, we have $A_t^{\top} \bvec{\tilde \mu}(t) = 0$. 

	Now, we get $\mathcal{R}_t(\bvec{\Phi}) \sim \mathcal{N} (0, A_t^{\top} \bvec{\tilde \Sigma}(t) A_t)$ and
	\[
 \expt[\mathcal{R}_t(\bvec{\Phi})] \leq \expt\left[\abs{\mathcal{R}_t(\bvec{\Phi})}\right] =  \sqrt{\frac{2}{\pi} A_t^{\top} \bvec{\tilde \Sigma}(t) A_t}. \]
	Note that $A_t^{\top} \bvec{\tilde \Sigma}(t) A_t$ is weighted summation of eigenvalues of $\bvec{\tilde \Sigma}(t)$. At any time $t$ in the coverage phase of the $k$-th epoch, $\max_{i\in V}\tr\big(\bvec{\tilde{\Sigma}}_i(t) \leq \alpha^{k}\tau$, and its follows that the summation of eigenvalues of $\bvec{\tilde \Sigma}(t)$ equals $\tr(\bvec{\tilde \Sigma}(t)) \leq \abs{V}\alpha^{k}\tau$. Thus,
	\[
 \expt \Bigg [\sum_{t\in \supscr{\mc T_k}{cov}}  \mathcal{R}_t(\bvec{\Phi}) \Bigg] \leq c_3 (\beta \sqrt{\alpha})^{k},  
 \]
	for some constant $c_3>0$, where $\supscr{\mc T_k}{cov}$ are the time slots in the coverage phase of the $k$-th epoch after the convergence of the coverage algorithm. We also use the fact that $|\supscr{\mc T_k}{cov}| \le \lceil \beta^k \rceil$.

	\noindent \textbf{Step 4 (Summary):} Summing up the expected coverage regret from the above steps, the expected cumulative coverage regret at the end of the $\ell$-th epoch $T_\ell$ satisfies
	\begin{align}\label{eq:regret-bound}
		\expt \Bigg[ \sum_{t=1}^{T_\ell} \mathcal{R}_t(\bvec{\Phi})\Bigg] \leq C_1 \ell +C_2 s_\ell +  \sum_{k=1}^\ell c_3 (\beta\sqrt{\alpha})^{k}, 
	\end{align}
	where $C_1, C_2 >0$ are some constants. 
	Since the number of epochs up to time $T$ satisfies $\ell \in O(\log T)$, substituting $\alpha= \beta^{-2/3}$ into~\eqref{eq:regret-bound}, and summing the resulting geometric series yields the claimed bound.
\end{proof}

\begin{remark}[Extensions of the proposed algorithm]\label{rem:extensions}
The proposed DSMLC algorithm serves as a template algorithm applicable to various environments and coverage algorithms, particularly in the single-task context, where numerous coverage algorithms have been developed. In the preliminary version of this work~\cite{LW-AM-VS:20h}, we established similar guarantees for the gossip coverage algorithm and graph environments. For continuous environments, while the non-parametric estimation framework can be immediately extended, the coverage algorithms may not converge in finite time~\cite{Cortes2004}. Nevertheless, our regret analysis can be easily adapted to continuous domains in scenarios where the coverage algorithm converges exponentially, such as in one-dimensional coverage problems~\cite{du2006convergence}. Additionally, multi-fidelity sensing can be incorporated similarly to~\cite{AM-LW-VS:21b, wei2020expedited}.
\end{remark}


\section{Simulation Study}
\label{sec:simulations}

We evaluate the proposed federated multitask coverage architecture in a heterogeneous firefighting scenario.
The environment is modeled as a $21\times 21$ connected grid graph.

We consider $N=9$ robots and $M=2$ task types: (i) Task 1: monitoring, and (ii) Task 2: fire suppression. The true sensory demand fields $\phi^j(v)$ are generated as mixtures of Gaussian kernels over the grid and normalized.  
The resulting spatial distributions are shown in Fig.~\ref{fig:true-demand-two}.




\begin{figure}[ht!]
\centering
\includegraphics[width=\linewidth]{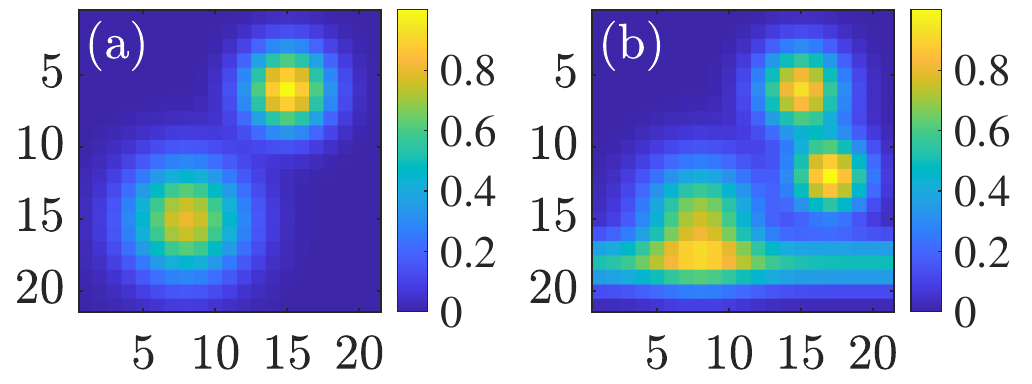}
\caption{True demand fields: (a) $\phi^1$ and (b) $\phi^2$.}
\label{fig:true-demand-two}
\end{figure}

Robot-task heterogeneity is encoded through the cost functions $f_i^j(d) = a_{ij}\, d$, which correspond to the linear distance-dependent cost of service. The coefficients $a_{ij} > 0$ represent task-dependent robot effectiveness and are generated randomly in each experiment. For Task~1 (monitoring), coefficients are generated as $a_{i1}=\max(0.25,\,1.0+0.2\xi_i)$ with $\xi_i\sim\mathcal{N}(0,1)$. 
For Task~2 (fire suppression), $6$ robots without dedicated firefighting capability use 
$a_{i2}=\max(0.25,\,2.3+0.25\xi_i)$, 
while $3$ firefighting-capable robots use 
$a_{i2}=\max(0.25,\,1.5+0.25\xi_i)$. 
This construction captures a higher baseline suppression cost for most robots while assigning a lower cost to specialized firefighting units.

\noindent
\textbf{Federated Multitask Coverage:} 
We first illustrate Algorithm~\ref{algo:FMC} under known sensory demand functions.  
Fig.~\ref{fig:federated-known} shows the evolution of the coverage cost, while Fig.~\ref{fig:two-snapshot} shows the optimal deployment and partitions to which the algorithm converges. 

\begin{figure}[ht!]
\centering
\includegraphics[width=\linewidth]{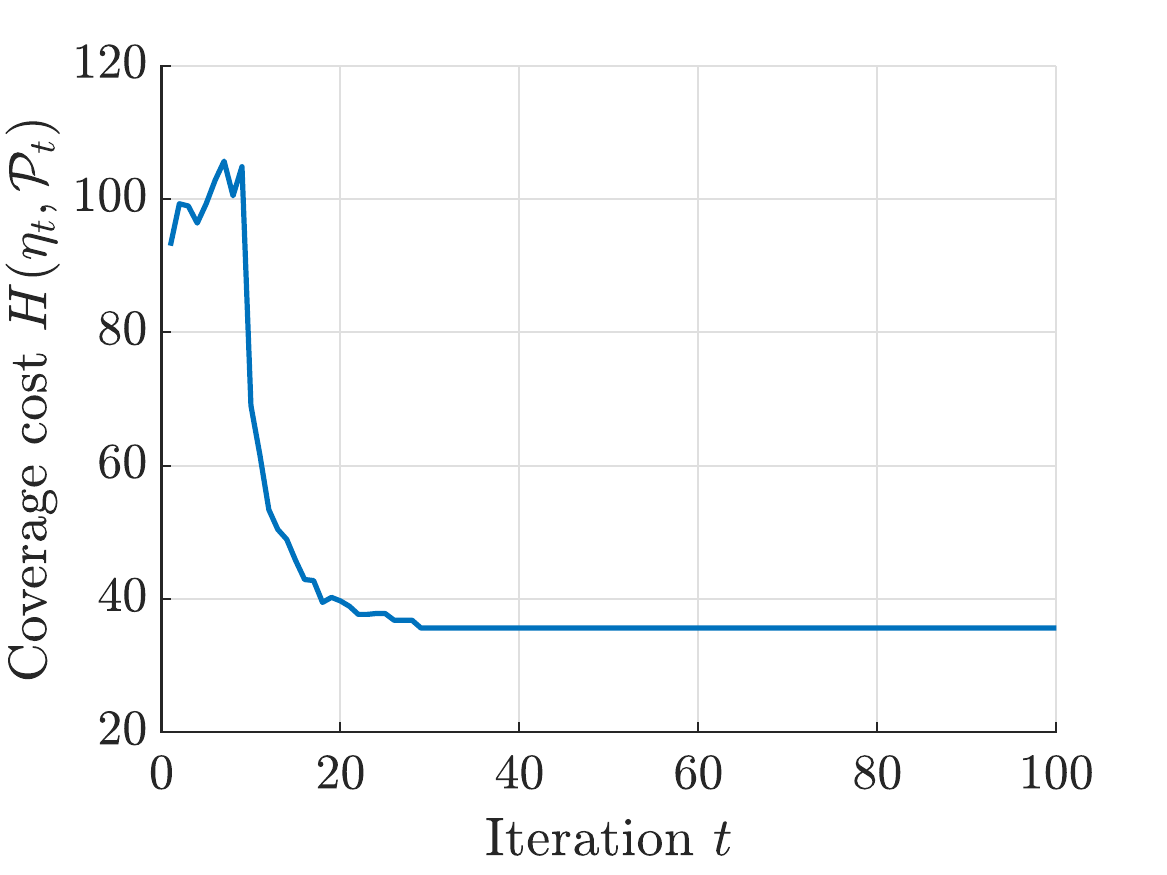}
\caption{Federated multitask coverage (FMTC) under known demand.}
\label{fig:federated-known}
\end{figure}

\begin{figure}[ht!]
\centering
\includegraphics[width=\linewidth]{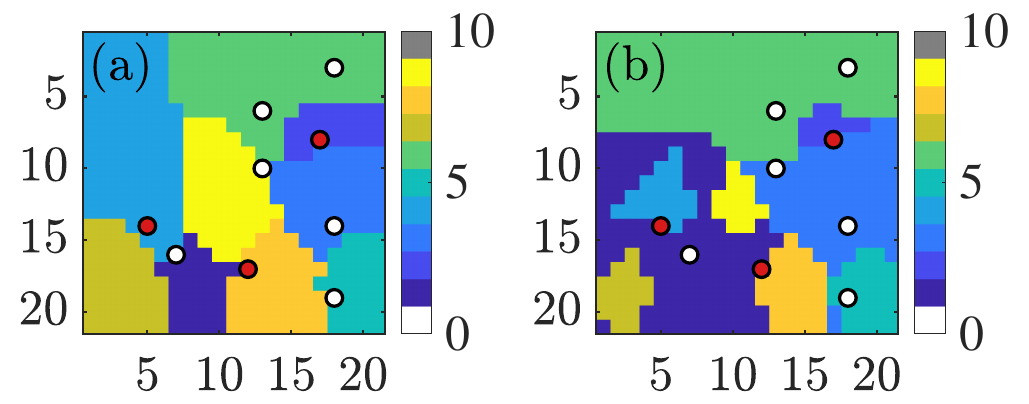}
\caption{Final deployment and multitask partition in the heterogeneous case. The red dots represent robots with superior Task 2 capability, and white dots represent other robots. The color of the partition is mapped to the robot identity with 0 and 10 indicating uncovered and multiply covered vertices, respectively. Robots with superior Task 2 capabilities have indices $\{1, 3, 6\}$. }
\label{fig:two-snapshot}
\end{figure}

The resulting deployment reflects both spatial demand intensity and robot-task heterogeneity. 
Robots with larger $a_{i1}$ concentrate near peaks of $\phi^1$, while robots with larger $a_{i2}$ distribute over vertices where $\phi^2$ dominates.

\noindent
\textbf{Single Task Learning and Coverage:} We now focus on a single task learning and coverage scenario: Task~1  with demand function $\phi^1$. We model the sensory field as a Gaussian process with a squared exponential kernel with $\sigma_v^2 = 1.0$ and $l = 0.18$. The measurement noise is taken as $\sigma = 0.2$. We compare the proposed DSMLC algorithm with a randomized multitask learning and coverage (RMLC) algorithm, which is an adaptation of the algorithm proposed in~\cite{Todescato2017}. DSMLC uses $\alpha =0.5$ and follows the epoch-based deterministic sequencing structure with robots communicating with the base station in a round-robin fashion. 

In the RMLC algorithm as well, the agents communicate with the base station in a round-robin fashion. At each communication round, robots transmit the measurements collected since the last contact, and the base station updates the GP and executes a single step of federated multitask coverage using the updated posterior mean as the demand function. At each time, each robot chooses to stay at or move to $\eta_i$ suggested by the base station with probability $1-p_i$, or to visit the most uncertain point in $\mc P^1_i$ and collect a measurement, where
\(
p_i = \frac{M_i}{M_i+\kappa},  \kappa=0.1,
\)
where $M_i$ is the maximum posterior variance within its assigned region.

Fig. ~\ref{fig:single-regret} shows cumulative regret for both algorithms. Both algorithms asymptotically converge to deployments consistent with the known-demand solution. DSMLC exhibits smaller regret due to coordinated sampling.


\begin{figure}[ht!]
\centering
\includegraphics[width=0.8\linewidth]{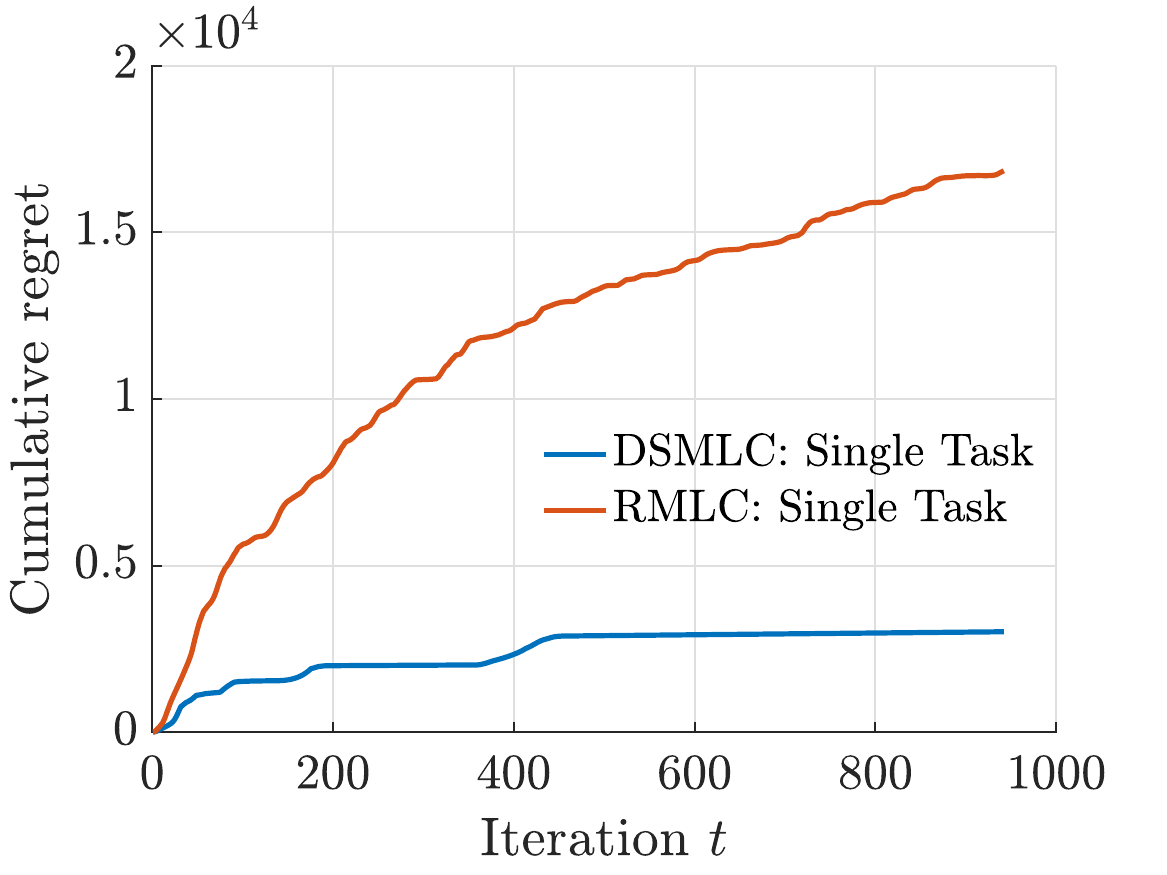}
\caption{Cumulative regret comparison in the single-task case.}
\label{fig:single-regret}
\end{figure}

\noindent
\textbf{Two Task Learning and Coverage:}
We next consider the full heterogeneous multitask setting.
The measurement noise parameter is taken as $\sigma = 0.2$, and the inter-task correlation is $0.65$. The most uncertain point in RMLC is selected by comparing the trace of the posterior covariance in  $P^1_i \union  P_i^2$. 
Cumulative regret is shown in Fig.~\ref{fig:two-regret}. The regret increases relative to the single-task case, reflecting the additional complexity of learning task-specific spatial structure under heterogeneity.

\begin{figure}[ht!]
\centering
\includegraphics[width=0.8\linewidth]{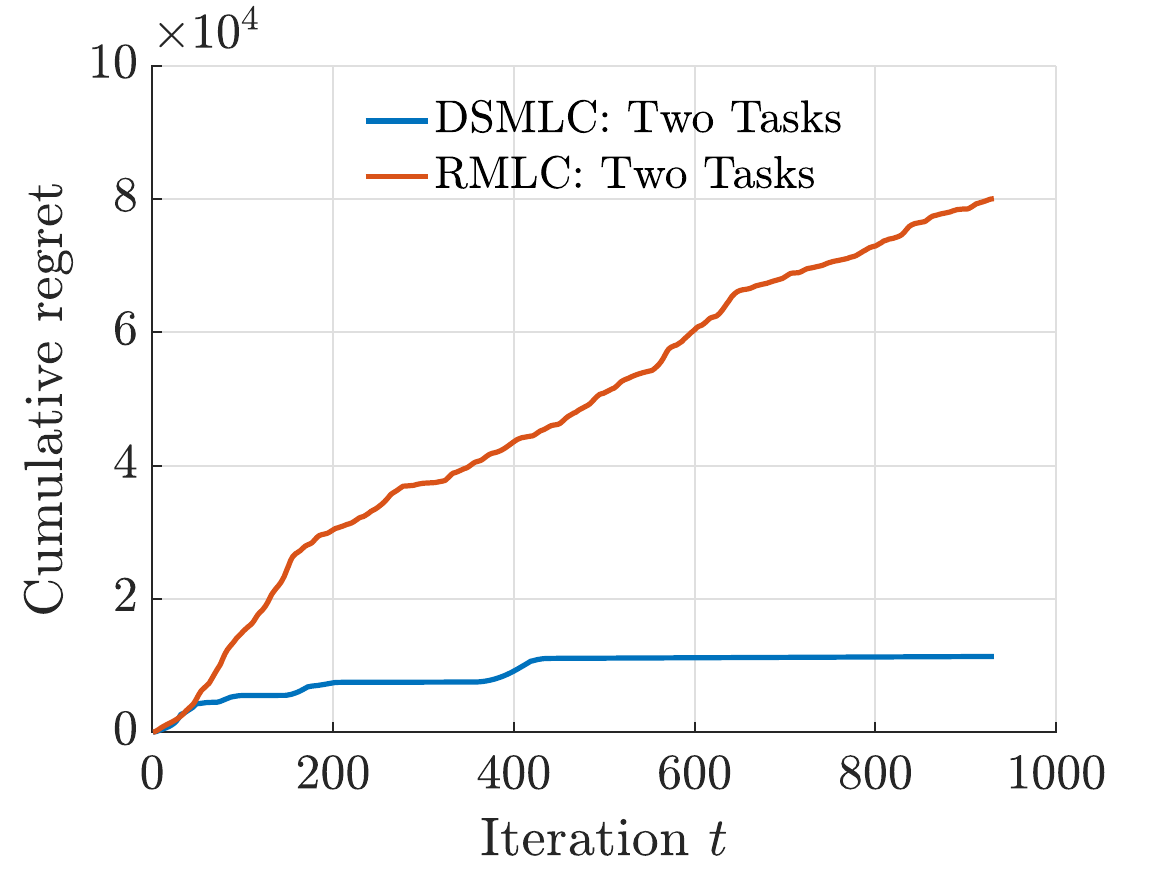}
\caption{Cumulative regret comparison in the heterogeneous two-task case.}
\label{fig:two-regret}
\end{figure}

While DSMLC achieves superior regret compared with RMLC in cases shown here. The performance is sensitive to the choice of $p_i$. Optimal choice of $p_i$ or the associated parameters in RMLC is an open area of investigation.


\section{Conclusions and Future Directions}
\label{sec:conclusions}
In this paper, we introduced a novel multitask coverage problem and designed algorithms to address it for both known and unknown sensory demands. We developed a federated multitask coverage algorithm for known demands and established its convergence properties.  We extended the algorithm to design an adaptive multitask coverage algorithm for unknown demands using a multitask Gaussian Process (GP) framework. We defined a new notion of multitask coverage regret and demonstrated that our adaptive algorithm achieves sublinear cumulative regret. We illustrated the empirical promise of our approach through numerical simulations.

Future research directions include extending the approach to settings where agents have unknown dynamics. Exploring nonstationary environments where the sensory field evolves over time is another promising avenue of research. Finally, studying the notion of social fairness within the proposed framework will also be an interesting direction.

\footnotesize 

\bibliographystyle{ieeetr}
\bibliography{coverage}

\end{document}